\documentclass[runningheads]{llncs}
\usepackage[T1]{fontenc}
\usepackage{graphicx}
\usepackage{mathtools}
\usepackage{amsfonts}
\usepackage{stmaryrd}

\usepackage{caption}
\usepackage{subfigure}

\newcommand{\metis}{\mbox{\sc Me$\!$T$\!$iS}}    
\newcommand{\scotch}{\mbox{\sc Scotch}}               
\newcommand{\patoh}{\mbox{\sc PaToH}}

\DeclareMathOperator{\CT}{{\cal T}} 
\DeclareMathOperator{\CS}{{\cal S}} 
\DeclareMathOperator{\CC}{{\cal C}} 
\DeclareMathOperator{\CN}{{\cal N}} 
\DeclareMathOperator{\CU}{{\cal U}} 
\DeclareMathOperator{\co}{\it p}        
\DeclareMathOperator{\we}{\it m}        

\begin{document}
\title{A Bi-Criteria FPTAS for Scheduling with Memory Constraints on Graphs with Bounded Tree-width}
\titlerunning{A Bi-Criteria FPTAS for Scheduling with Memory Constraints}
%
%
\author{Eric Angel\inst{1}  \and S\'ebastien Morais\inst{2,3} \and Damien Regnault\inst{1}}%
\authorrunning{E. Angel et al.}
%
%
\institute{
IBISC, Univ Evry, Universite Paris-Saclay, 91025, Evry, France\\ \email{\{Eric.Angel,Damien.Regnault\}@univ-evry.fr},
\and
CEA, DAM, DIF, F-91297 Arpajon, France\\ \email{Sebastien.Morais@cea.fr}
\and 
LIHPC - Laboratoire en Informatique Haute Performance pour le Calcul et la simulation - DAM \^Ile-de-France, University of Paris-Saclay}
\maketitle              
\begin{abstract}
In this paper we study a scheduling problem arising from executing numerical simulations on HPC architectures. With a constant number of parallel machines, the objective is to minimize the makespan under memory constraints for the machines. Those constraints come from a neighborhood graph $G$ for the jobs.  Motivated by a previous result on graphs $G$ with bounded path-width, our focus is on the case when the neighborhood graph $G$ has bounded tree-width. Our result is a bi-criteria fully polynomial time approximation algorithm based on a dynamic programming algorithm. It allows to find a solution within a factor of $1+\epsilon$ of the optimal makespan, where the memory capacity of the machines may be exceeded by a factor at most $1+\epsilon$. This result relies on the use of a nice tree decomposition of $G$ and its traversal in a specific way which may be useful on its own. The case of unrelated machines is also tractable with minor modifications.
\end{abstract}

\section{Introduction}
In this paper, we study the scheduling problem $Pk|G, mem|C_{max}$ previously introduced in~\cite{Angel2016} where the number of machines is a fixed constant. This problem is motivated by running distributed numerical simulations based on high-ordered finite elements or volume methods~\cite{Ern_Guermond-FEM-2004,LeVeque_2002}. Such approaches require the geometric domain of study to be discretized into basic elements, called cells, which form a mesh. Each cell has a computational cost, and a memory weight depending on the amount of data (i.e. density, pressure, \ldots) stored on that cell. Moreover, performing the computation of a cell requires, in addition to its data, data located in its  neighborhood\footnote{The neighborhood is most of the time topologically defined (cells sharing an edge or a face).}. For a distributed simulation, the problem is to assign all the computations to processing units with bounded memory capacities, while minimizing the makespan. As an illustration of the previous notions, let us consider Figure~\ref{fig_mesh} where a mesh and its associated computations are assigned onto 3 processing units. Each color corresponds to a processing unit and the total amount of memory needed by each processing unit is not limited to the colored cells but extends to some adjacent cells. An exploded view of the mesh is pictured in Figure~\ref{fig_mesh_ghost_cells}, where we consider an edge-based adjacency relationship and where the memory needed by each processing unit is equal to both colored and white cells. In practice, efficient partitioning tools such as \scotch~\cite{pellegrini96}, \metis~\cite{metis98}, Zoltan~\cite{zoltan99} or \patoh~\cite{patoh11} are used. However, the solutions returned by these tools may not respect the memory capacities of the processing units~\cite{MPUMC20}.

\begin{figure}[ht]
\centering
\begin{subfigure}[]{\centering \label{fig_mesh}
\includegraphics[width=.3\linewidth, height=2cm]{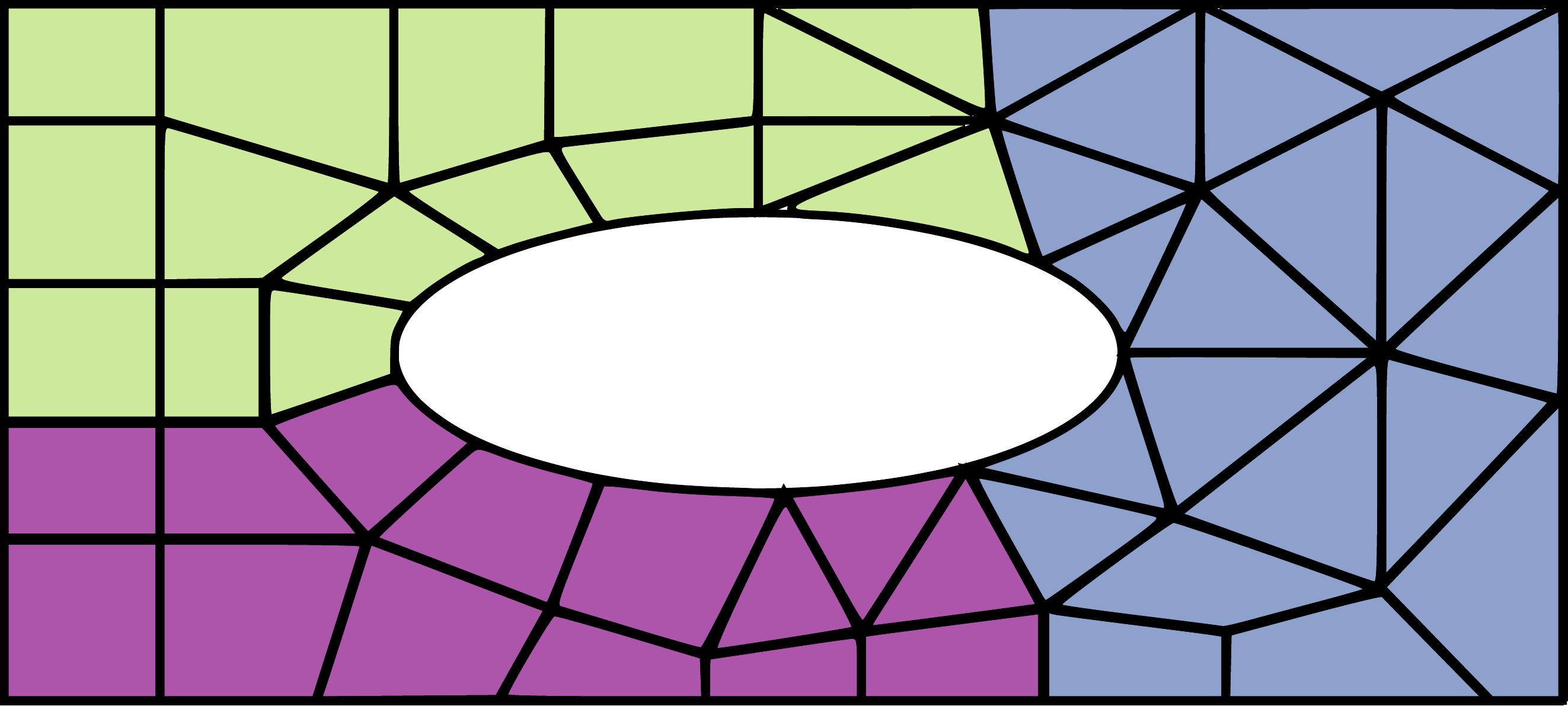}}
\end{subfigure} ~ ~ ~ 
\begin{subfigure}[]{\centering \label{fig_mesh_ghost_cells}
\includegraphics[width=.35\linewidth, height=2cm]{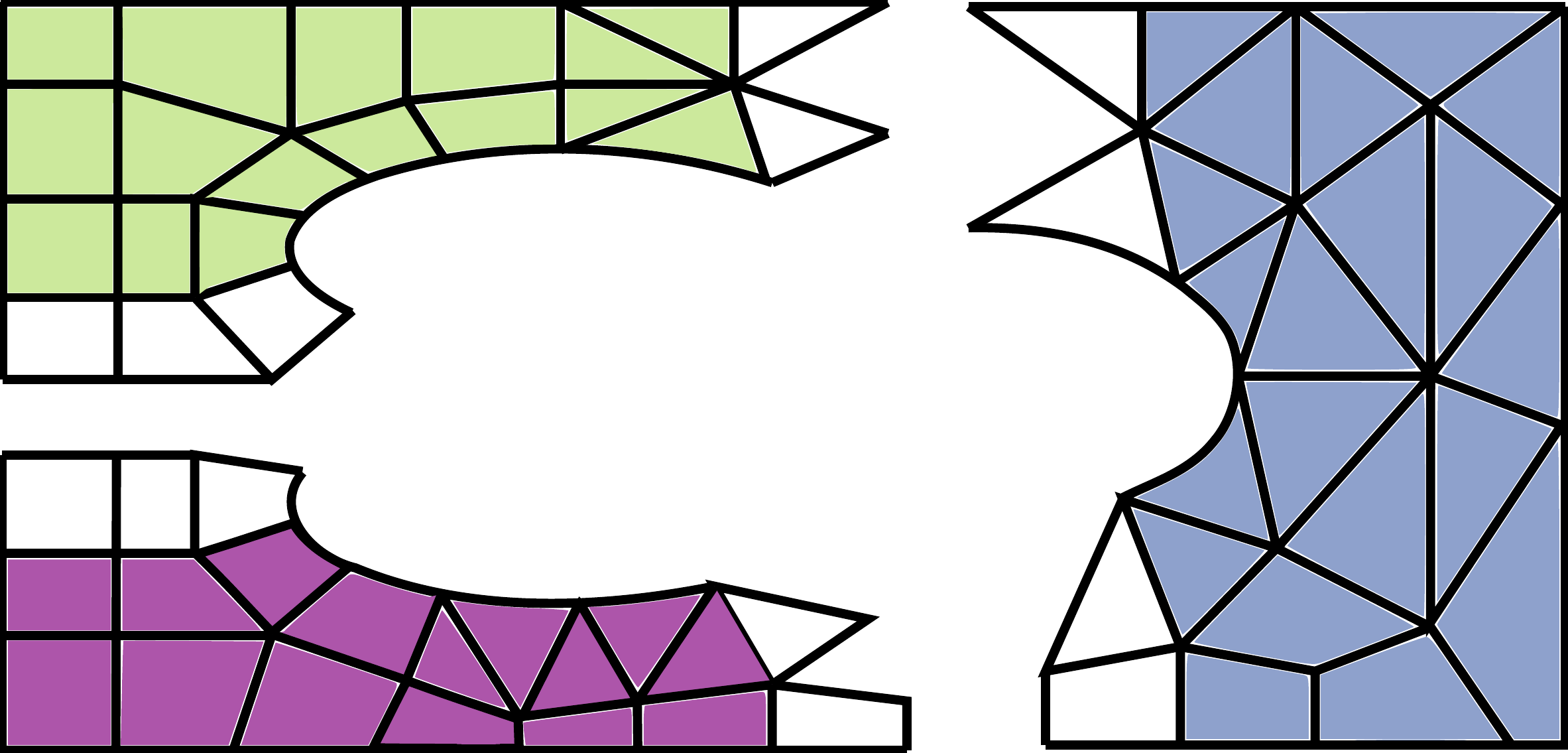}}
\end{subfigure}
\caption{\label{fig:info-intro}In ($a$), a 2D mesh and its computations are assigned onto 3 processing units. In ($b$), an exploded view of the assignment with an edge-based adjacency relationship.}
\end{figure}

Formally, the scheduling problem under memory constraints is defined as follows. We have a set of $n$ jobs $J$, and each job $j\in J$ requires $\co_j\in \mathbb{N}$ {\it units of time} to be executed (computation time) and an amount $\we_j\in \mathbb{N}$ of memory. 
Jobs have to be assigned among a fixed number $k$ of identical {\it machines}, each machine $l$ having a memory capacity $M_l\in \mathbb{N}$, for $l=1,\ldots,k$.  
Additionally we have an undirected graph $G(J,E)$, which we refer to as the {\it neighborhood graph}. Two jobs $j \in J$ and $j' \in J$ are said to be adjacent
if there is an edge $(j,j') \in E$ in $G$.
Moreover, each job $j$ requires data from its set of {\it adjacent jobs}, denoted by $\mathcal{N}(j) := \{j'\in J\,|\,(j,j')\in E\}$. 
For a subset of jobs $J'\subseteq J$, we note $\CN(J') := \cup_{j\in J'} \CN(j) \setminus J'$ and denote $\CN[J'] := \CN(J') \cup J'$. When a subset of jobs $J'\subseteq J$ is scheduled on a machine, this machine needs to allocate an amount of memory equal to 
$\sum_{j \in {\cal N}[J']} \we_j$, while its processing time is $\sum_{j\in J'} \co_j$. The objective is to assign each job of $J$ onto exactly a machine, such that the makespan (the maximum processing time over all machines) is minimized and ensuring  that the amount of memory allocated by each machine is smaller than or equal to its memory capacity. 

The scheduling problem under memory constraints embraces other well-known {\bf NP}-hard scheduling problems and cannot be solved in polynomial time unless {\bf P=NP}. Thus, one could be interested in developping approximation algorithms. An {\it $\alpha$-approximation algorithm} (for some $\alpha \geq 1$) for a minimization problem is a polynomial-time algorithm that produces, for any given problem instance $I$, a solution whose value is at most $\alpha$ times the optimum value. In particular, a {\it fully polynomial-time approximation scheme} (FPTAS) is a family of (1+$\epsilon$)-approximation algorithms for all $\epsilon > 0$ whose time complexity is polynomial in both the input size and $1 / \epsilon$. Considering the scheduling problem under memory constraints, one could wonder if approximation algorithms can be obtained when the memory constraints are relaxed. For $\alpha \geq 1$ and $\beta \geq 1$, an $(\alpha, \beta)$ {\it (bi-criteria) approximation algorithm} returns a schedule with objective value at most $\alpha C$ and with memory load at most $\beta M$, where $C$ and $M$, respectively, are the maximum computation time and the memory load of an optimal schedule with respect to the makespan. Following the terminology of~\cite{biFPTAS}, a bi-FPTAS 
for the scheduling problem under memory constraints is a FPTAS which is a bicriteria $(1+\epsilon, 1+\epsilon)$ {\it approximation algorithm}.

\subsection{Related problems} \label{subsec-relat-prob}
The problem $Rk|G,mem|C_{max}$ contains other well-known {\bf NP}-hard scheduling problems. When $\we_j=0$ for each job $j$, the problem $Rk|G,mem|C_{max}$ becomes the scheduling problem $Rk||C_{max}$ for which several approximations algorithms exist~\cite{Gairing:2007:FCA:1244475.1244730,Lenstra:1990:AAS:81018.81019,woeginger2000does}. When the neighborhood graph has no edges, and the memory is bounded on each machine, and $\we_j=1$ for each job $j$, we get the so-called Scheduling Machines with Capacity Constraints problem (SMCC). In this problem, each machine can process at most a fixed number of jobs. 
Zhang et al.~\cite{Zhang:2009:ASM:1574693.1574730} gave a $3$-approximation algorithm by using the iterative rounding method.
Saha and Srinivasan~\cite{Saha10anew} gave a $2$-approximation in a more general scheduling setting, i.e. Scheduling Unrelated Machines with Capacity Constraints. 
Lately, Keller and Kotov~\cite{DBLP:journals/orl/KellererK11} gave a $1.5$-approximation algorithm.
 Chen et al. established an EPTAS~\cite{DBLP:conf/cocoa/0011JLZ16} for this problem and, 
for the special case of two machines, Woeginger designed a FPTAS~\cite{Woeginger:2005:CST:2307218.2307260}.

\subsection{Main Contribution}
As the scheduling problem under memory constraints is a generalization of those well-known scheduling problems, a reasonable question is to know whether we can get approximation algorithms, which could depend on some parameters of the neighborhood graph, when the number of machines is a fixed constant. We answered this question in a previous paper~\cite{Angel2016} by providing a {\it fixed-parameter tractable} (FPT) algorithm with respect to the path-width of the neighborhood graph, which returns a solution within a ratio of $(1+\varepsilon)$ for both the optimum makespan and the memory capacity constraints (assuming that there exists at least one feasible solution). In this paper we extend this result by providing a bi-FPTAS for graphs with tree-width bounded by a constant. Unlike the FPT algorithm which relies on the numbering of the vertices of the neighborhood graph, the bi-FPTAS takes advantage of a nice tree decomposition of the neighborhood graph and of its traversal in a particular way to bound the algorithm complexity. 

\subsection{Outline of the Paper}
We start by briefly recalling in Section~\ref{sec-def} the definitions of different notions useful in the sequel. We then provide in Section~\ref{sec-dyna} an algorithm that computes all the solutions to this problem. This algorithm consists of three steps: build a nice tree decomposition of $G(J,E)$; compute a layout $L$ defining a bottom-up traversal of the nice tree decomposition; and use a dynamic programming algorithm traversing the nice tree decomposition following $L$. Since the time complexity of this algorithm is not polynomial in the input size, we apply the Trimming-of-the-State-Space technique \cite{Ibarra:1975:FAA:321906.321909} in Section~\ref{sec-trimming} obtaining a bi-criteria approximation algorithm for graphs with tree-width bounded by a constant. Finally, we give some concluding remarks in  Section~\ref{sec-conclu}.

\section{Definitions} \label{sec-def}
Throughout this paper we consider simple, finite undirected graphs. Let us start by defining the notions of tree decomposition, tree-width and nice tree decomposition. The notions of tree decomposition and tree-width were initially introduced in the framework of graph minor theory~\cite{Robertson198339}. 
For a graph $G(J,E)$, let $J(G):=J$ be its vertices and $E(G):=E$ be its edges. A {\it tree decomposition} for $G$ is a pair $(T,X)$, where $T:=(J(T), E(T))$ is a tree, and 
$X:=(X_u)_{u \in J(T)}$ is a family of subsets of $J$ satisfying the following conditions:
\begin{enumerate}
\item For each $j \in J(G)$ there is at least one $u \in J(T)$ such that $j \in X_u$.
\item For each $\{j,j'\} \in E(G)$ there is at least one $u \in J(T)$ such that $j$ and $j'$ are in $X_u$.
\item For each $j \in J(G)$, the set of vertices $u \in J(T)$ such that $j \in X_u$ induces a subtree of $T$.
\end{enumerate}

To distinguish between vertices of $G$ and $T$, the latter are called {\it nodes}. The {\it width} of a tree decomposition is $\max(|X_u| - 1: u \in J(T))$ and the {\it tree-width} of $G$, noted $tw(G)$, is the minimum width over all tree decompositions of $G$. A graph $G(J,E)$ is illustrated on Figure~\ref{subfig_neighb_graph} and a tree decomposition of this graph is illustrated on Figure~\ref{subfig_tree_decomposition}. 
\begin{figure}[ht]
\centering
\begin{subfigure}[A graph $G(J,E)$.
]{
\centering
\includegraphics[width=.25\linewidth, height=2.5cm]{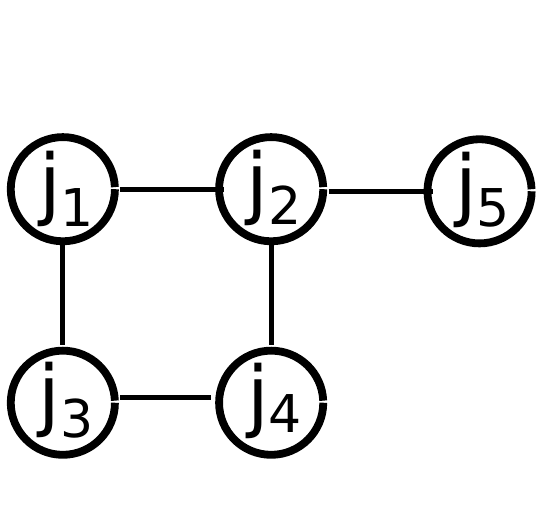}
\label{subfig_neighb_graph}
}
\end{subfigure}
\begin{subfigure}[A tree decomposition $(T,X)$ of the graph $G(J,E)$.]{
\centering
\includegraphics[width=.4\linewidth, height=2.7cm]{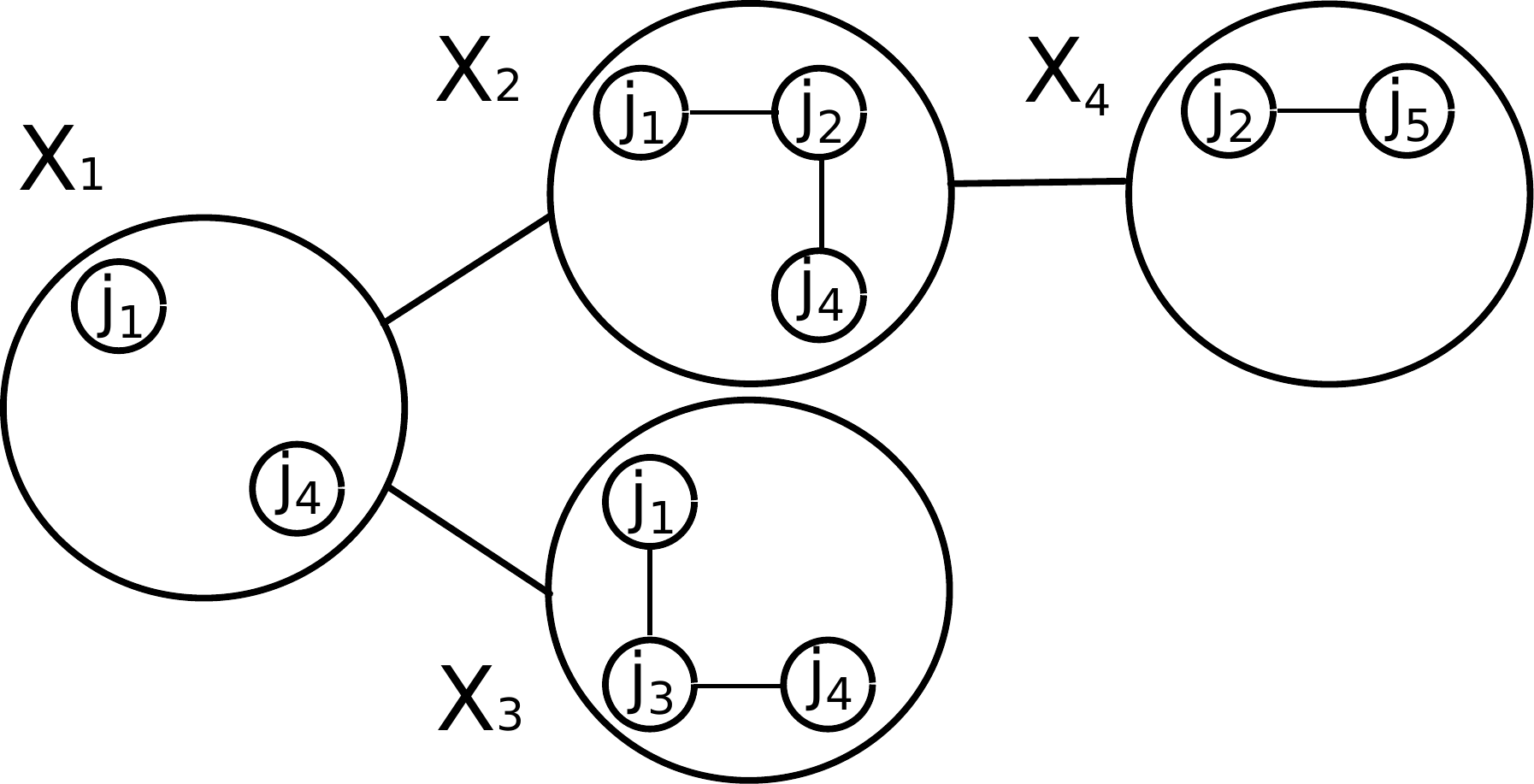}
\label{subfig_tree_decomposition}
}
\end{subfigure}
\caption{Example of a graph $G(J,E)$ in (a) and a tree decomposition $(T,X)$ of this graph where $X$ is composed of the 
sets 
$X_1 = \{j_1, j_4\}$, $X_2 = \{j_1, j_2, j_4\}$, $X_3 = \{j_1, j_3, j_4\}$, $X_4 = \{j_2, j_5\}$ in (b).
}
\label{fig_tree_decomposition}
\end{figure}


Choosing an arbitrary node $r \in J(T)$ as root, we can make a {\it rooted tree decomposition} out of $(T, X)$ with natural parent-child and ancestor-descendant relations. A node without children is called a {\it leaf}.

A rooted tree decomposition $(T,X)$ with root $r$ is called {\it nice} if every node $u \in J(T)$ is of one of the following types:
\begin{itemize}
\item \textbf{Leaf:} node $u$ is a leaf of $T$ and $|X_u|=1$.
\item \textbf{Introduce:} node $u$ has only one child $c$ and there is a vertex $j \in J(G)$ such that $X_u = X_c \cup \{j\}$.
\item \textbf{Forget:} node $u$ has only one child $c$ and there is a vertex $j \in J(G)$ such that $X_c = X_u \cup \{j\}$.
\item \textbf{Join:} node $u$ has only two children $l$ and $r$ such that $X_u = X_l = X_r$.
\end{itemize}

In Figure~\ref{fig_nice_tree_decomposition} we present a nice tree decomposition of $G(J,E)$ illustrated on Figure~\ref{subfig_neighb_graph}
\begin{figure}[ht!]
\centering
\includegraphics[width=\linewidth, height=3.9cm]{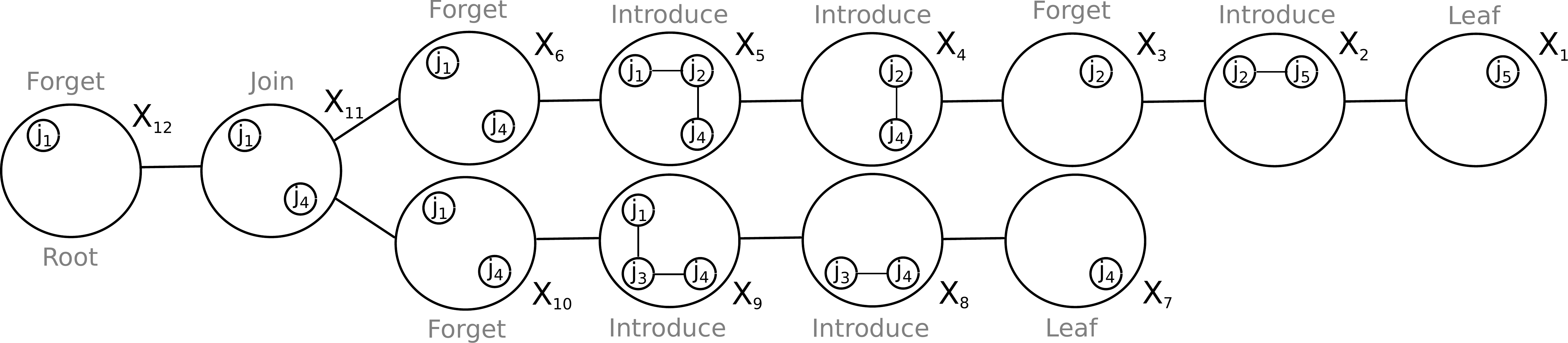}
\caption{Example of a nice tree decomposition of the graph $G(J,E)$ with width $tw(G)=2$ where the node types are written in grey.}
\label{fig_nice_tree_decomposition}
\end{figure}

Note that a vertex of $J(G)$ can be forgotten at most once in a node of $J(T)$. Otherwise, it would conflict with the third condition listed in the definition of a tree decomposition. We leverage this property later in the article.

There is an alternative definition to the nice tree decomposition where the root $r$ and all leaves $u$ of $T$ are such that $X_r$ = $X_u$ = $\emptyset$. But one can switch from one of these decompositions to the other in a trivial way.

When $G$ is a graph with $tw(G)=h$, where $h$ is any fixed constant, we can compute a tree decomposition of $G$ in linear time with tree-width at most $h$~\cite{Bodlaender96alinear}.  Given a tree decomposition $(T,X)$ of $G(J,E)$ of constant width $h \geq 1$, there is an algorithm that converts it into a nice tree decomposition $(T', X')$ with the same width $h$ and with at most $4n$ nodes, where $n=|J(G)|$, in $O(n)$ times~(Lemma 13.1.3 in \cite{Kloks1994TreewidthCA}). In the rest of the article, we will consider a nice tree decomposition obtained in this way.

%
%

Now, let us introduce the notion of {\it layout} of a nice tree decomposition $(T, X)$, which is simply a one-to-one mapping $L: J(T) \rightarrow \{1, \ldots, |J(T)|\}$. We say that a layout $L$ defines a {\it bottom-up traversal} of a nice tree decomposition $(T, X)$ if for any edge $\{u,v\}\in E(T)$ such that $v$ is a child of $u$ one has $L(v)<L(u)$. In that case, we say that $L$ is a bottom-up layout. 

%

\section{An Exact Algorithm Using Dynamic Programming}\label{sec-dyna}
Briefly, our algorithm consists of three steps. First, we build a nice tree decomposition $(T, X)$ of the graph $G(J,E)$ with bounded tree-width. Such a tree decomposition can be obtained in polynomial time for graph $G$ with tree-width bounded by a constant (see Section~\ref{sec-def}). Then, we compute a specific layout $L$ defining a bottom-up traversal of the nice tree decomposition. Finally, a dynamic programming algorithm passes through the nodes following the previously defined order $L$ and computes a set $\CS_{L(u)}$ of states, which encodes partial solutions for $G_i=(J_i,E_i)$ a subgraph of $G=(J,E)$, for each node $u \in J(T)$. In Section~\ref{sec-31}, we start by presenting the dynamic programming algorithm where we detail how the set of states $\CS_{L(u)}$ is computed depending on the type of node $u$. Then, in Section~\ref{sec-32}, we give a proof of correctness of our dynamic programming algorithm when the nodes of the nice tree decomposition are traversed in a bottom-up way.
Eventually, we compute the complexity of our dynamic programming algorithm when the decomposition is traversed following the layout $L$. This layout is used to bound the complexity of our algorithm and, being bottom-up, it is compliant with the pre-requisite on proof of completeness.



\subsection{The Dynamic Programming Algorithm \label{sec-31}}
The presentation of the dynamic algorithm is done for two machines, but it can be generalized to a constant number $k$ of machines, with $k > 2$. The dynamic algorithm goes through $|J(T)|$ phases. Each phase $i$, with $i = 1, \ldots, |J(T)|$, processes the node $L^{-1}(i) \in J(T)$ and produces a set $\CS_i$ of states. 
In the sequel, for sake of readability, we use the notation $Z_i := X_{L^{-1}(i)}$.
Each state in the state space $\CS_i$ encodes a solution for the graph $G_i=(J_i, E_i)$, where $J_i := \cup_{o=1}^{i} Z_o$ with $J_0 = \emptyset$, and $E_i := E_{i-1} \cup E_{Z_i}$ with $E_{0}=\emptyset$ and $E_{Z_i}$ the set of all edges in $E$ which have both endpoints in $Z_i$.

For each phase $i$, we denote by $J_L(i)$ the set of vertices of $J(G)$ which have not been forgotten when going through nodes $L^{-1}(1)$ to $L^{-1}(i)$. For convenience, we note $J_L(0):=\emptyset$. 
Formally, $J_L(i) := J_i \setminus V_R(i)$,
where $V_R(i)$ is the set of vertices that where removed in a Forget node  $o$ such that $L(o) \leq i$.

A state $s \in \CS_i$ is a vector $[c_1, c_2, c_3, c_4, \CC_i]$ where:
\begin{itemize}
\item $c_1$ (resp. $c_2$) is the total processing time on the first (resp. second) machine in the constructed schedule,
\item $c_3$ (resp. $c_4$) is the total amount of memory required by the first (resp. second) machine in the constructed schedule,
\item $\CC_i$ is an additional structure, called {\it combinatorial frontier}. For a given solution of $G_i(J_i, E_i)$, it is defined as $\CC_i:=(J_L(i), \sigma_i, \sigma_i')$ where $\sigma_i : J_L(i) \rightarrow \{1, 2\}$ and $\sigma_i' : J_L(i)\rightarrow \{0,1\}$ such that $\sigma_i(j)$ is the machine on which $j\in J_L(i)$ has been assigned, and $\sigma_i'(j):=1$ if the machine on which $j$ is not assigned, i.e. machine $3 - \sigma_i(j)$, has already memorised the data of $j$. Notice that $J_L(i)\subseteq J_i$ and keeping into memory the combinatorial frontier with respect to $J_L(i)$ rather than $J_i$ is a key point in our algorithm in order to bound its complexity.
\end{itemize}

In the following, we present how to compute $\CS_{i}$ from $\CS_{i-1}$ depending on the type of node $L^{-1}(i)$. For that, we present how states of $\CS_i$ are obtained from an arbitrary state $s = [c_1, c_2, c_3, c_4, \CC_{i-1}] \in \CS_{i-1}$. When $L^{-1}(i)$ is a Leaf node with $Z_i=\{j\}$ or an Introduce node with $j$ the vertex introduced, we note $s_a$ ($a=1,2$) the state of $\CS_i$ obtained from $s$ and resulting from the assignment of $j$ to machine $a$, and $\CC_i^a$ the combinatorial frontier obtained from $\CC_{i-1}$ when $j$ is assigned to machine~$a$.

\smallskip
{\bf Leaf} Let $L^{-1}(i) \in J(T)$ be a Leaf of $T$ with $Z_i = \{j\}$. For each state of $\CS_{i-1}$ we add at most two states in $\CS_i$. If $j \in J_L(i-1)$, it means that $j$ has already been assigned to a machine. Therefore, there is nothing to do and $\CS_i = \CS_{i-1}$. Now, let us assume that $j \notin J_L(i-1)$.
In this case, we must compute two new states taking into account the assignment of $j$ to machine one or two. 
We have 
$$ s_a = [c_1 + \delta_{a,1} c_j,~ c_2 + \delta_{a,2} c_j,~ c_3 + \delta_{a,1} \we_j,~ c_4 + \delta_{a,2} \we_j,~ \CC_i^a] $$
where $\delta$ is the Kronecker function ($\delta_{i,j}=1$ if $i=j$, and $\delta_{i,j}=0$ otherwise). Since $j \notin  J_L(i-1)$, the new combinatorial frontier is obtained by extending $\CC_{i-1}$ in adding new information related to $j$, i.e. $\sigma_i(j)=a$ and $\sigma_i'(j)=0$. 
Note that we have $\sigma_{i}'(j) = 0$ because $j$ was not assigned before phase $i$ and $E_{Z_i} = \emptyset$.


\smallskip
{\bf Introduce} Let $L^{-1}(i) \in J(T)$ be an Introduce node of $T$ and $j \in J(G)$ being the vertex introduced. Again, for each state $s$ of $\CS_{i-1}$ we are going to add at most two states to $\CS_i$ depending on $j$ assignment. However, processing an Introduce node differs from a Leaf because we may have to consider new edges. This happens when $E_i \backslash E_{i-1} \neq \emptyset$. There are two cases to consider. The first one is when $j \in J_L(i-1)$. In that case, job $j$ has already been assigned on machine $a=\sigma_{i-1}(j)$. We add a state in $\CS_i$ for every state $s$ in $\CS_{i-1}$. Let $F_a$ and $F_a'$ be the set of edges such that
\begin{equation}
\label{eq_f}
F_a = \{ \{j,j'\} \in E_{Z_i}: a \neq \sigma_{i-1}(j') \text{ and } \sigma_{i-1}'(j')=0\},
\end{equation}
\begin{equation}
\label{eq_f_prime}
F_a' = \{ \{j,j'\} \in E_{Z_i}: a \neq \sigma_{i-1}(j') \text{ and } \sigma_{i-1}'(j)=0\}.
\end{equation}
The set $F_a$ represents the new edges in $E_{Z_i}$ inducing additional amount of data on machine $a$. The set $F_a'$ represents the new edges in $E_{Z_i}$ inducing that $\we_j$ must be added on the machine not processing $j$. Note that some edges in $E_{Z_i}$ may have already been considered in a previous node and that they can't be a part of $F_a$ or $F_a'$. Thus, we have
$$ s_a = [c_1,~ c_2,~ c_3 + \delta_{a,1} \alpha_i^1 + \delta_{a,2} \beta_i^1,~ c_4 + \delta_{a,2} \alpha_i^2 + \delta_{a,1} \beta_i^2,~ \CC_i^a] $$ where 
$\alpha_i^a = \sum_{\{j,j'\} \in F_a} \we_{j'}$ and $\beta_i^a = \we_j I\llbracket F_a' \neq \emptyset \rrbracket$
 where $I\llbracket A \rrbracket$ is the indicator function which returns one if condition $A$ is satisfied and zero otherwise. 
Finally, the combinatorial frontier of the new state $s_a$ is obtained from that of $s$ by updating, if necessary, the information of $j$ and vertices $j'$ such that $\{j, j'\} \in F_a$. If we have $F_a' \neq \emptyset$, it means that $j$ was not memorised by machine $3 - a$ in state $s$. However, this is no longer the case for $s_a$ as new edges have been taken into account leading us to $\sigma_i'(j) = 1 \neq \sigma_{i-1}'(j)$. If we have $F_a \neq \emptyset$, then some vertices processed by machine $3-a$ were not memorised by machine $a$ in state $s$. Again, this is no longer the case in $s_a$ following the inclusion of new edges leading us to $\sigma_i'(j') = 1 \neq \sigma_{i-1}'(j')$ for every vertex $j'$ such that $\{j, j'\} \in F_a$.

Now, if $j \notin J_L(i-1)$ then we add two states in $\CS_i$ for every state $s \in \CS_{i-1}$. For $a=1,2$, we have
\begin{equation*}
\begin{aligned}
s_a = [ {} & c_1 + \delta_{a,1} \co_j,~ c_2+ \delta_{a,2} \co_j, \\
{}& c_3 + \delta_{a,1} ( \we_j + \alpha_i^1) + \delta_{a,2} \beta_i^1,~ c_4 + \delta_{a,2} (\we_j + \alpha_i^2) + \delta_{a,1} \beta_i^2, ~ \CC_i^a].
\end{aligned}
\end{equation*}

The way to obtain the first four coordinates of each new state in $\CS_i$ is similar to the case where $j \in J_L(i-1)$ except that we have to add $\co_j$ and $\we_j$ on the machine processing $j$. In the case of the combinatorial frontier, updates defined for $j \in J_L(i-1)$ also apply and we have to add information related to $j$ since it was unknown so far. The added data is $\sigma_i(j) = a$ and $\sigma_i'(j) =  I\llbracket \exists j' \in Z_i : \{j, j'\} \in E_{Z_i} \text{ and }\sigma_{i-1}(j') \neq a \rrbracket$.

\smallskip
{\bf Forget} Let $L^{-1}(i) \in J(T)$ be a Forget node of $T$ and $j \in J(G)$ being the vertex forgotten. This type of node is easier to handle than previous ones since we don't have to deal with new vertex or edges. The only thing to do is to withdraw $j$ from the combination frontier. Thus, for each state $s \in \CS_{i-1}$ we add a state $s' \in \CS_i$ where the combinatorial frontier of $s'$ is equal to that of $s$ from which information on $j$ was removed.


\smallskip {\bf Join}
Let $L^{-1}(i) \in J(T)$ be a Join node of $T$. This type of node is even simpler to deal with than the previous one. Once again, there are no new vertex or edges to handle. Moreover, we don't forget any vertex. For each state $s \in \CS_{i-1}$ we add $s$ to $\CS_i$. Thus, we have $\CS_{i} = \CS_{i-1}$.
\smallskip

Our algorithm ends up by returning the state $s=[c_1,c_2,c_3,c_4,\CC_{|J(T)|}] \in \CS_{|J(T)|}$ with $c_3 \leq M_1$, $c_4 \leq M_2$ and such that $\max\{c_1,c_2\}$ is minimum.

\subsection{Algorithm Correctness \label{sec-32}}
Now, let us present the proof of correctness of our dynamic programming algorithm when the nodes of the nice tree decomposition are traversed in bottom-up. We will prove our algorithm correctness by maintaining the following invariant: the states in $\CS_i$ encode all the solutions for the graph $G_i =(J_i, E_i)$, defined at Section~\ref{sec-31}.

\subsubsection*{Initialization} Let us start with the first node encountered. Let $G_0=(J_0,E_0)$ be an empty graph and $\CS_0$ be the set composed of the single state $[0,0,0,0,\CC_0]$ where $\CC_0$ does not store information. The nodes being traversed in bottom-up, the first node encountered is a Leaf. Let $j \in J(G)$ be the vertex such that $Z_1 = \{j\}$. Since $j \notin J_L(0)$ we have $\CS_1 = [ (\co_j, 0, \we_j, 0, \CC_1^1), (0, \co_j, 0, \we_j, \CC_1^2)]$ where, for $a=1,2$, $\CC_1^a$ is such that $\sigma_1(j)=a$ and $\sigma_1'(j)=0$. These two states encode the assignment of $j$ on machines one and two when considering the graph $G_1=(J_1, E_1)$. Moreover, the combinatorial frontier obtained allows us to keep in memory potentially necessary knowledge for graphs of which $G_1=(J_1,E_1)$ is a sub-graph. Thus the invariant is correct for the first node.

\subsubsection*{Maintenance} Now let us assume that the invariant holds for $L^{-1}(i-1) \in J(T)$ and let us prove that it is still correct for $L^{-1}(i) \in J(T)$.

{\bf Leaf} Let $L^{-1}(i) \in J(T)$ being a Leaf with $Z_i = \{j\}$. If $j \in J_L(i-1)$ then our algorithm states that $\CS_i = \CS_{i-1}$. In that case, the invariant holds because $G_{i}=(J_{i},E_{i})$ is equal to $G_{i-1}=(J_{i-1},E_{i-1})$. Now, if $j \notin J_L(i)$ then our algorithm adds two new states in $\CS_i$ for every state in $s \in \CS_{i-1}$ to take into account the assignment of $j$ to machine one and two. Each new state is obtained by adding $\co_j$ and $\we_j$ according to the assignment of $j$ and the associated combinatorial frontier is obtained by extending the combinatorial frontier of $s$ with information on $j$ assignment, i.e. $\sigma_i(j)=a$ and $\sigma_i'(j)=0$. Since we are dealing with a Leaf and $j \notin J_L(i)$ we have $G_{i}=(V_{i-1} \cup \{j\},E_{i-1})$. Therefore, the invariant holds.
\smallskip

{\bf Introduce} Let $L^{-1}(i) \in J(T)$ being an Introduce node with $j \in J(G)$ being the vertex introduced. If $j \in J_L(i-1)$ then our algorithm adds one new state in $\CS_i$ for every state in $\CS_{i-1}$. A new state in $\CS_{i}$ is obtained from a state in $\CS_{i-1}$ by adding, if needed, some amount of data on machine one and two.
Let $a = \sigma_{i-1}(j)$ and $F_a$ and $F_a'$ be the sets defined in~(\ref{eq_f}) and (\ref{eq_f_prime}). We note $F_a''$ the set such that $F_a'' = E_{Z_i} \backslash (F_a \cup F_a')$.


\begin{lemma}
Let $s$ be a state encoding a solution of a graph $G'=(J',E')$. Then, if we add an edge $e=\{j,j'\}$ such that $j \in J'$, $j' \in J'$ and $e \in F_a''$ then $s$ also encodes a solution of the graph  $G'=(J', E' \cup e)$.
\label{lemma_graph_encoding}
\end{lemma}

\begin{proof}
The proof of this lemma is based on the fact that introducing such edge does not make $s$ inconsistent with graph $G'=(J', E' \cup e)$. Let us begin by noting that adding an edge $e=\{j, j'\} \in F_a''$ does not require to modify the processing times in $s$ to make it a state encoding a solution of $G'=(J', E' \cup e)$.
Indeed, since $s$ encodes a solution for $G'=(J',E')$, the processing time induced by the assignment of $j$ and $j'$ has already been encoded. Now, suppose that $e \in F_a''$. Then, we have either $j$ and $j'$ that are assigned to the same machine, or $j$ and $j'$ that are memorised by both machines. In either case, adding such an edge does not require to modify the amount of memory or combinatorial frontier in
$s$ to make it a state encoding a solution of $G'=(J', E' \cup e)$. {\hfill $\square$ \bigbreak}
\end{proof}

Let us now go back to our algorithm. On the machine processing $j$, our algorithm adds $\we_{j'}$ for every vertex $j' \in J_i$ such that $\{j,j'\} \in F_a$. Indeed, since $j'$ is on a different machine than $j$ and that this machine does not memorise $j'$, it is necessary to add $\we_{j'}$ on machine $\sigma_{i-1}(j)$ to take into account the edge $\{j, j'\}$. On the machine not processing $j$, our algorithm adds $\we_{j}$ if there is an edge $\{j, j'\} \in F_a'$. Indeed, as $j'$ is on a different machine than $j'$ and $j$ is not memorised by this machine, it is necessary to add $\we_{j}$ on machine $\sigma_{i-1}(j')$ to take into account the existence of such an edge. Finally, we update the combinatorial frontier information on vertex $j$ if $F_a' \neq \emptyset$ and on vertices $j'$ such that $\{j, j'\} \in F_a$. Therefore, the states returned by our algorithm encode solutions for the graph $G'=(J_i, E_{i-1} \cup F_a \cup F_a' \cup F_a'')$ and the combinatorial frontier is consistent with the addition of new vertices or edges. 
According to Lemma~\ref{lemma_graph_encoding}, our algorithm encodes solutions for the graph $G_i = (J_i, E_{i})$ since $E_i = E_{i-1} \cup E_{Z_i}$ and $E_{Z_i}\!= F_a \cup F_a' \cup F_a''$. Thus, the invariant holds.

Now, if $j \notin J_L(i-1)$ the proof of the invariant enforcement is similar to the case where $j \in J_L(i-1)$. The difference lies in the fact that $j$ is not yet assigned. Thus, one must generate two new states in $\CS_i$ for each state in $\CS_{i-1}$ and the processing time, and amount of memory, of $j$ must be added on the machine processing $j$.
\smallskip

{\bf Forget} Let $L^{-1}(i) \in J(T)$ be a Forget node of $T$ and $j \in J(G)$ being the vertex forgotten. Here, our algorithm generates the states of $\CS_i$ by taking those of $\CS_{i-1}$ from which it removes information on vertex $j$ from the combinatorial frontier. First, let us note that $G_i=(J_i,E_i)$ is equal to $G_{i-1}=(J_{i-1},E_{i-1})$ and the invariant holds. Notice that since we traverse $T$ in bottom-up, we know that removing a vertex $j$ implies that all edges linked to it have been explored. Otherwise, it would lead to the violation of a property of the tree decomposition (the third listed in Section~\ref{sec-def}). Therefore, we can stop memorising the information related to vertex $j$.
\smallskip

{\bf Join} Let $L^{-1}(i) \in J(T)$ be a Join node of $T$. In that case, our algorithm computes $\CS_i$ by retrieving the states of $\CS_{i-1}$ without modifying them. Since we have $G_i=(J_i,E_i)$ equal to $G_{i-1}=(J_{i-1},E_{i-1})$ and no modification on the combinatorial frontier is performed, the invariant holds.

\subsubsection*{Termination} Finally, from the first and second conditions listed in the definition of the tree decomposition, we know that the graph $G_{|J(T)|}=(J_{|J(T)|}, E_{|J(T)|})$ is equal to $G=(J,E)$. Since our invariant is valid for the first node and during the transition from nodes $L^{-1}(i-1)$ to $L^{-1}(i)$, our algorithm returns an optimal solution for the scheduling problem under memory constraints.

\subsection{Algorithm Complexity \label{sec-33}}
Let us now evaluate the time complexity of our dynamic programming algorithm. Let $J_L^{max} := \max_{1 \leq i \leq |J(T)|} |J_L(i)|$.
Let $\co_{sum} := \sum_{j \in J(G)} \co_j$ and $\we_{sum} := \sum_{j \in J(G)} \we_j$, then for each state $s = [c_1, c_2, c_3, c_4, \CC_i] \in \CS_i$, $c_1$ and $c_2$ are integers between $0$ and $\co_{sum}$, $c_3$ and $c_4$ are integers between $0$ and $\we_{sum}$. The number of distinct combinatorial frontiers is $4^{J_L^{max}}$. Therefore, the number of states is $|\CS_i| = O(\co_{sum}^2 \times \we_{sum}^2 \times 4^{J_L^{max}})$. The dynamic programming algorithm processes all $|J(T)| = O(n)$ nodes of the nice tree decomposition. Each state in a phase can give at most two states in the next phase with a processing time of $O(J_L^{max})$ to compute these states. 
Recall also that in the algorithm, if two states $s$ and $s'$ have the same components, including the same combinatorial frontier, then only one of them is kept in the state space. The time complexity to test whether two states $s$ and $s'$ are the same is thus proportional to the length of the combinatorial frontier, and is therefore $O(J_L^{max})$.
We obtain that the overall complexity of the dynamic programming algorithm is 
$O(n \times |\CS_i| \times (J_L^{max} + |\CS_i| J_L^{max}))=
O(n \times J_L^{max} \times (\co_{sum}^2 \times \we_{sum}^2 \times 4^{J_L^{max}})^2)$.

Notice that $J_L^{max}$ depends on the chosen layout $L$, and to minimize this complexity it is therefore important to find a layout $L$ with a small $J_L^{max}$. 

\begin{lemma}
\label{lemma_size_combinatorial_frontier}
There exists a bottum-up layout $L$ of the nice tree decomposition such that 
$J_L^{max} \leq  tw(G)\, \lceil \log 4n\rceil $.
\end{lemma}

\begin{proof}

To prove that such a layout exists we present an algorithm which, when applied to the root of the nice tree decomposition, computes a bottom-up layout $L$ such that $J_L^{max} \leq  tw(G)\, \lceil \log 4n\rceil $. To ease the understanding of certain parts of the proof, these parts will be illustrated on Figure~\ref{fig:lemma2} where a tree with 174 nodes is depicted.

The algorithm works as follows. We perform a depth-first search starting from the root node, and when we have a Join node we first go to the subtree having the greatest number of nodes. With this depth-first search we get a discovery and finishing times for each node. The labeling is obtained by sorting the nodes in increasing order of their finishing time. As an example, the nodes of the nice tree  decomposition in Figure~\ref{fig_nice_tree_decomposition} 
have been labelled according to this procedure if we consider in this example that $L^{-1}(i)=i$ for $1\leq i\leq 12$.

Now, let us analyze $J_L^{max}$ on the layout returned by our algorithm. Recall that we use the notation $Z_i := X_{L^{-1}(i)}$ and let us define the operator $\sqcup$ such that $Z_i\sqcup Z_{i+1} := Z_i \setminus \{j\}$ if $L^{-1}(i+1)$ is a Forget node, with $j$ the vertex forgotten, and $Z_i\sqcup Z_{i+1} := Z_i \cup Z_{i+1}$ otherwise. Notice that $J_L(i) = \sqcup_{o\leq i} Z_o$ and that if we have a set of consecutive nodes $L^{-1}(l)$, $L^{-1}(l+1)$, \ldots, $L^{-1}(u)$ such that $L^{-1}(i+1)$ is a parent node for $L^{-1}(i)$ $(l\leq i\leq u-1)$, then $\sqcup_{i=l}^u Z_i = Z_u$. Moreover if this chain is maximal, i.e. $L^{-1}(u+1)$ is not a parent node of $L^{-1}(u)$, then it means that the parent of $L^{-1}(u)$ is a Join node. For any node $L^{-1}(i)$, we have $J_L(i) = \sqcup_{o\leq i} Z_o = \cup_{l\in A}  X_l$, with $A$ a set of nodes, of minimum size, that we call {\it critical}. This set of critical nodes $A$ can be obtained by taking the last node in each maximal chain over nodes $L^{-1}(1)$ to $L^{-1}(o)$. Thus, $A$ is composed of the current node $L^{-1}(i)$ along with other nodes whose parents are Join nodes. Such set $A$ is illustrated in Figure~\ref{fig:lemma2}(a) where we consider $i=166$ and where the nodes composing $A$ are green colored.

For a Join node $L^{-1}(i)$ having two childrens $L^{-1}(l)$ and $L^{-1}(r)$, let denote by $T_l(i)$ and $T_r(i)$ the corresponding subtrees. We will assume that
$|T_l(i)|\geq |T_r(i)|$ and therefore during the depth-first search we use, node $L^{-1}(l)$ will be examined before node $L^{-1}(r)$. We say that $L^{-1}(l)$ (resp. $L^{-1}(r)$) is the left (resp. right) children of $L^{-1}(i)$. By the way the depth-first search is performed, all nodes in $A$, excepted the current node $L^{-1}(i)$, are left children of Join nodes, and these Join nodes are on the path $P$ between the root node and the current node $L^{-1}(i)$. In Figure~\ref{fig:lemma2}(a), such path $P$ contains the Join nodes purple colored.

Now, let us bound the number of Join nodes on the path $P$. First, we construct a reduced graph by removing the nodes of $R$ where $R$ is the set of nodes of $P$ that are not Join nodes. Such a reduced graph is illustrated in Figure~\ref{fig:lemma2}(b). By doing this set of deletions, we get a tree with fewer than $4n$ nodes (recall that the nice tree decomposition we started from has at most $4n$ nodes). The number of Join nodes is equal to the length of the reduced path $P\setminus R$ which is $\lceil \log 4n \rceil$. Indeed, starting from the root, each time we go on a node along this path the number of remaining nodes is divided by at least 2.

Thus, we have proved that $|A|\leq \lceil \log 4n\rceil$ for any node $L^{-1}(i)$ labelled with our algorithm. Recall that $J_L(i) = \cup_{l\in A}  X_l$,
and moreover from the definition of tree-width, we have $|X_l| \leq tw(G)$. Thus, we have $|J_L(i)| \leq  tw(G)\, \lceil \log 4n\rceil$ and the proof is complete. {\hfill $\square$ \bigbreak}

\end{proof}

\begin{figure}[ht]
\begin{subfigure}[
]{
\includegraphics[width=.5\linewidth, height=4.5cm]{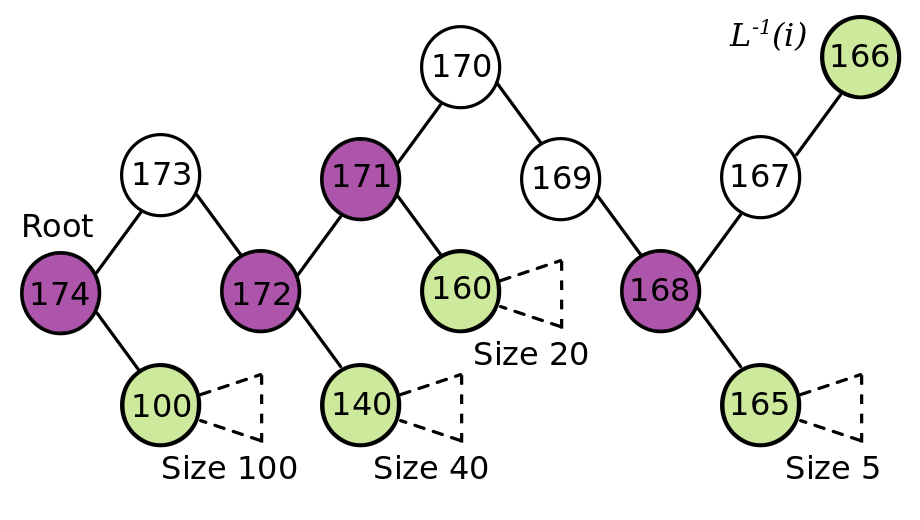}
}
\end{subfigure}
\begin{subfigure}[]{
\includegraphics[width=.35\linewidth, height=4.5cm]{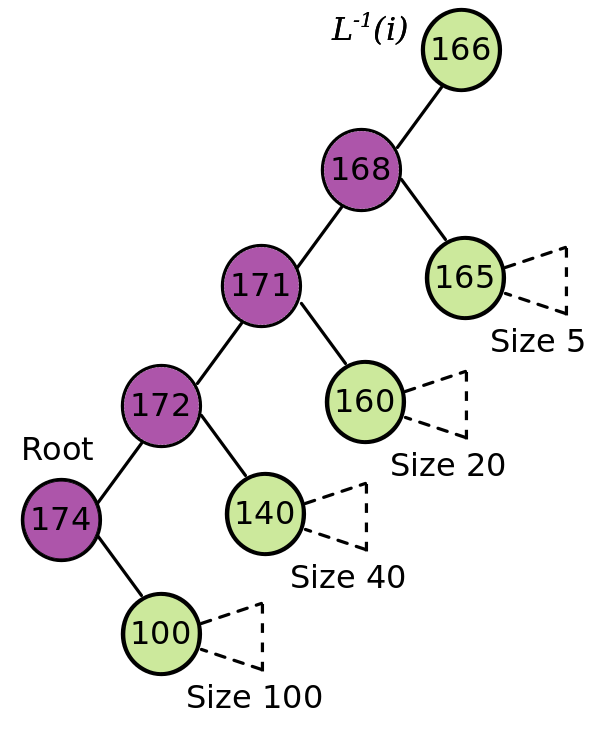}
}
\end{subfigure}
\caption{Illustration of the proof of Lemma~\ref{lemma_size_combinatorial_frontier} on a possible tree with 174 nodes. The tree is labelled with a bottum-up layout $L$, and for notational convenience we consider that $L^{-1}(i)=i$ for $1\leq i\leq 174$. Some subtrees are represented by dashed triangles. On Figure~(a) is depicted the tree. When considering node $166$, the set of critical nodes $A=\{166, 165, 160, 140, 100\}$. All nodes in $A$, excepted the node $166$, are left children of Join nodes, and these Join nodes are on the path $P = \{174,173,172,171,170,169,168,167,166\}$ between the root node 174 and the node $166$. On Figure~(b) is depicted the reduced tree obtained by removing all nodes in $P$ which are not Join nodes, namely $R$.
In each figure, the set of critical nodes $A$ associated to node $166$ is green colored and the Join nodes in $P$ are purple colored.
}
\label{fig:lemma2}
\end{figure}

Using the previous defined layout, we obtain an overall complexity of our dynamic programming algorithm of $O(\co_{sum}^4 \times\we_{sum}^4\times tw(G) \times \log(n) \times n^{2tw(G)+1}\times {16}^{tw(G)})$. The time complexity of this dynamic programming algorithm being pseudo-polynomial (because of $\co_{sum}$ and $\we_{sum}$), we are going to transform it into a bi-FPTAS.

\section{Application of a trimming technique}\label{sec-trimming}
In this Section, we propose a bi-FPTAS derived from the algorithm presented in Section~\ref{sec-dyna}. 
To transform the dynamic programming algorithm, we apply an approach for transforming a dynamic programming formulation into a FPTAS. This approach, called the {\it trimming-the-state-space} technique is due to Ibarra \& Kim~\cite{Ibarra:1975:FAA:321906.321909} and consists in iteratively thin out the state space of the dynamic program by collapsing states that are close to each other.

In the approximation algorithm, we are going to trim the state space by discarding states that are close to each other. While carrying these states deletions, we must ensure that the resulting errors cannot propagate in an uncotrolled way. To this end, we characterize a notion of proximity between states. We define $\Delta := 1 + \varepsilon/8n$, with $\varepsilon>0$ a fixed constant. Let us first consider the first two coordinates of a state $s=[c_1,c_2,c_3,c_4,\CC_i]$. We have $0\leq c_1\leq \co_{sum}$ and $0\leq c_2\leq \co_{sum}$. We divide each of those intervals into intervals of the form $[0]$ and $[\Delta^l,\Delta^{l+1}]$, with $l$ an integer value getting from $0$ to $L_1 := \lceil \log_\Delta (\co_{sum})\rceil =
 \lceil ln (\co_{sum})/ ln (\Delta) \rceil \leq 
\lceil (1+\frac{8n}{\varepsilon}) ln (\co_{sum}) \rceil$. In the same way, we divide the next two coordinates into
intervals of the form $[0]$ and $[\Delta^l,\Delta^{l+1}]$, with $l$ an integer value getting from $0$ to $L_2 := \lceil \log_\Delta (\we_{sum})\rceil$. The union of those intervals defines a set of non-overlapping boxes. If two states have the same combinatorial frontier and have their first four coordinates falling into the same box, then they encode similar solutions and we consider them to be close to each other.


The approximation algorithm proceeds in the same way as the exact algorithm, except that we add a trimming step to thin out each state space $\CS_i$. The trimming step consists in keeping only one solution per box and per combinatorial frontier. Thus, the worst time complexity of this approximation algorithm is $O(L_{1}^4 \times L_{2}^{4} \times tw(G) \times \log(n) \times n^{2tw(G)+1}\times {16}^{tw(G)})$. We therefore get a bi-FPTAS when the tree-width $tw(G)$ is bounded by a constant.



\begin{theorem}\label{th-approx}
There exists a bi-FPTAS for the problem $Pk|G,mem|C_{max}$ when the tree-width of $G$ is bounded by a constant, which returns a solution within a ratio of $(1+\varepsilon)$ for the optimum makespan, where the memory capacity $M_i$, $1 \leq i \leq k$, of each machine may be exceeded by at most a factor $(1+\varepsilon)$.
\end{theorem}

For sake of readability, the proof is presented when $k=2$. In the conclusion, we mention the general case when $k$ is any fixed constant. We denote by $\CU_i$ (resp. $\CT_i$) the state space obtained before (resp. after) performing the trimming step at the $i$-th phase of the algorithm. The proof of this theorem relies on the following lemma.

\begin{lemma}
For each state $s=[c_1,c_2,c_3,c_4, \CC_i] \in \CS_i$, there exists a state $[c_1^\#,c_2^\#,c_3^\#,c_4^\#,\CC_i] \in \CT_i$ such that
\begin{equation}
c_1^\# \leq \Delta^i c_1\ \text{ and } \ c_2^\# \leq \Delta^i c_2 \ \text{ and } \ c_3^\# \leq \Delta^i c_3 \ \text{ and } \ c_4^\# \leq \Delta^i c_4.
\label{ineq_statement}
\end{equation}
\label{lemma_delta_power}
\end{lemma}

\begin{proof}

The proof of this lemma is by induction on $i$. The first node we consider is a Leaf of the nice tree decomposition and we have $\CT_1 = \CS_1$. Therefore, the statement is correct for $i=1$. Now, let us suppose that inequality~(\ref{ineq_statement}) is correct for any index $i-1$ and consider an arbitrary state $s = [c_1, c_2, c_3, c_4, \CC_i] \in \CS_i$. Due to a lack of space, proof of the validity of the Lemma when passing from phase $i-1$ to $i$ is only presented for a node of type Introduce. Note that the proof for other types of nodes can be derived from that of an Introduce node.
Let $L^{-1}(i)$ be an Introduce node with $j \in J(G)$ being the vertex introduced. We must distinguish between cases where $j$ belongs to $J_L(i-1)$ and where he does not.

First, let us assume that $j \in J_L(i-1)$. Then $s$ was obtained from a state $[w,x,y,z,\CC_{i-1}] \in \CS_{i-1}$ and $s=[w, x, y + \delta_{a,1}\alpha_i^1 + \delta_{a,2} \beta_i^1, z + \delta_{a,2}\alpha_i^2 + \delta_{a,1} \beta_i^2, \CC_i^a]$ with $a=\sigma_i(j)$. According to the induction hypothesis, there is a state $[w^\#, x^\#, y^\#, z^\#, \CC_{i-1}] \in \CT_{i-1}$ such that 

\begin{equation}
\label{eq_hypothese_req_1}
w^{\#} \leq \Delta^{i-1} w\ \text{,   } \ x^{\#} \leq \Delta^{i-1} x\ \text{,   } \ y^{\#} \leq \Delta^{i-1} y\ \text{,   } \ z^{\#} \leq \Delta^{i-1} z.
\end{equation}

The trimmed algorithm generates the state $[w^\#, x^\#, y^\# + \delta_{a,1}\alpha_i^1 + \delta_{a,2} \beta_i^1, z^\# + \delta_{a,2}\alpha_i^2 + \delta_{a,1} \beta_i^2, \CC_i^a] \in {\cal U}_i$ and may remove it during the trimming phase, but it must leave some state $t=[c_1^\#, c_2^\#, c_3^\#, c_4^\#, \CC_{i}^a] \in {\cal T}_i$ that is in the same box as $[w^\#, x^\#, y^\# + \delta_{a,1}\alpha_i^1 + \delta_{a,2} \beta_i^1 , z^\# + \delta_{a,2}\alpha_i^2 + \delta_{a,1} \beta_i^2, \CC_i^a] \in~{\cal U}_i$. This state $t$ is an approximation of $s$ in the sense of~(\ref{eq_hypothese_req_1}). 

Indeed, its first coordinate $c_1^\#$ satisfies 
\begin{equation}
c_1^\# \leq \Delta(w^\#) \leq \Delta(\Delta^{i-1}w) \leq \Delta^i w  = \Delta^i c_1,
\end{equation}
its third coordinate $c_3^\#$ satisfies 
\begin{equation}
\begin{split}
c_3^\# & \leq
\Delta(y^\# + \delta_{a,1} \alpha_i^{1} + \delta_{a,2} \beta_i^{1})   \leq \Delta(\Delta^{i-1}y + \delta_{a,1} \alpha_i^{1} + \delta_{a,2} \beta_i^{1}) \\
& \leq \Delta^{i} y + \Delta (\delta_{a,1} \alpha_i^{1} + \delta_{a,2} \beta_i^{1}) \leq \Delta^i  c_3
\end{split}
\end{equation}
and its last coordinate is the same as $s$. By similar arguments, we can show that $c_2^\# \leq \Delta^i c_2$ and $c_4^\# \leq \Delta^i c_4$.

Now, let us assume that $j \notin J_L(i-1)$.
In that case, the state $s$ was obtained from a state $[w, x, y, z, \CC_{i-1}] \in \CS_{i-1}$ and either 
$s = [w+\co_j, x, y + \we_j + \alpha_i^1, z + \beta_i^2,\CC_i^1]$ 
or 
$s = [w, x+\co_j, y + \beta_i^1, z + \we_j + \alpha_i^2, \CC_i^2]$.
We assume that $s = [w+\co_j, x, y + \we_j + \alpha_i^1, z + \beta_i^2,\CC_i^1]$ as, with similar arguments, the rest of the proof is also valid for the other case. By the inductive assumption, there exists a state $[w^\#, x^\#, y^\#, z^\#, \CC_{i-1}] \in \CT_{i-1}$ that respects (\ref{eq_hypothese_req_1}). The trimmed algorithm generates the state $[w^\# + \co_j, x^\#, y^\# + \we_j + \alpha_i^1, z + \beta_i^2,  \CC_i^1] \in {\cal U}_i$ and may remove it during the trimming phase. However, it must leave some state $t=[c_1^\#, c_2^\#, c_3^\#, c_4^\#, \CC_{i}^1] \in {\cal T}_i$ that is in the same box as$[w^\# + \co_j, x^\#, y^\# + \we_j + \alpha_i^1, z + \beta_i^2, \CC_i^1] \in {\cal U}_i$. This state $t$ is an approximation of $s$ in the sense of (\ref{eq_hypothese_req_1}). Indeed, its last coordinate $\CC_i^1$ is equal to $\CC_i$ and, by arguments similar to those presented for $j \in J_L(i-1)$, we can show that $c_o^\# \leq \Delta^i c_o$, for $o \in \llbracket 1,4\rrbracket$. Thus, our assumption is valid during the transition from phase $i-1$ to $i$ when $i$ is an Introduce node.

Since the proof for the other type of nodes can be derived from the proof of an Introduce node, the inductive proof is completed.

{\hfill $\square$ \bigbreak}
\end{proof}

Now, let us go back to the proof of Theorem~\ref{th-approx}. After at most $4n$ phases, the untrimmed algorithm outputs the state $s=[c_1,c_2,c_3,c_4,\CC]$ that minimizes the value $\max\{c_1,c_2\}$ such that $c_3 \leq M_1$ and $c_4 \leq M_2$. By Lemma~\ref{lemma_delta_power}, there exists a state $[c_1^\#,c_2^\#,c_3^\#,c_4^\#,\CC] \in {\cal T}_n$ whose coordinates are at most a factor of $\Delta^{4n}$ above the corresponding coordinates of $s$. Thus, we conclude that our trimmed algorithm returns a solution where the makespan is at most $\Delta^{4n}$ times the optimal solution and the amount of memory for each machine is at most $\Delta^{4n}$ its capacity. Moreover, since $\Delta := 1 + \varepsilon/8n$, we have $\Delta^{4n} \leq 1 + \varepsilon$ for $\varepsilon\leq 2$.

So we have presented an algorithm that returns a solution such that the makespan is at most $(1+\varepsilon)$ times the optimal solution and the amount of memory for each machine is at most $(1+\varepsilon)$ its capacity. It ends the proof of Theorem~\ref{th-approx}.

Notice that if we no more consider the maximum memory load as a constraint but as an objective,
our algorithm can be modified to get an $(1+\epsilon)$ approximate Pareto set \cite{approxPareto}. 
 At the end of the algorithm, we consider all the states in $\CS_{|J(T)|}$ and compute for each state $s=[c_1,c_2,c_3,c_4,\CC_{|J(T)|}] \in \CS_{|J(T)|}$ the vector $[\max\{c_1,c_2\},\max\{c_3,c_4\}]$. Then among this set of vectors, we compute the set of nondominated vectors. The problem of finding a set of nondominated vectors among a given set of vectors has been first introduced in \cite{maximavector} and there exists a lots of efficient algorithms for this problem (see for example \cite{maxima-recent}). It is easy to see that the obtained set is a $(1+\varepsilon)$-approximate Pareto set for our problem.

\section{Conclusion} \label{sec-conclu}
Given $2$ machines and a neighborhood graph of jobs with bounded tree-width, we have presented an algorithm that returns a solution, where the capacity of the machines may be exceeded by a factor at most $1+\epsilon$, if at least one solution exists for the scheduling problem under memory constraints. This algorithm consists of three steps: construct a nice tree decomposition of the neighborhood graph; compute a specific bottom-up layout $L$ of the nice tree decomposition; and use a transformed dynamic programming algorithm traversing the nice tree decomposition following $L$. The specific bottom-up layout $L$ is designed to bound the complexity of our algorithm but it is not optimal. It would be interesting to lower this complexity by taking into account the number of vertices associated to each node (see for example~\cite{LAM201163}) and avoiding counting duplicate vertices. However, using layout $L$, the output of our algorithm is generated in polynomial time and is such that the makespan is at most $(1+\varepsilon)$ times the optimal solution and the amount of memory for each machine is at most $(1+\varepsilon)$ its capacity. 

Although the algorithm is presented for 2 machines, it can be extended to any number of machines as adding machines means increasing the number of dimensions of a state. It would require to modify the combinatorial frontier such that $\sigma_i'(j)$ would be the machines on which $j$ has not been assigned and which have memorized the data of $j$. This would change the time complexity to 
$O(n \times L_1^{2k} \times L_2^{2k} \times k \times tw(G) \times log(n) \times n^{2(k tw(G) + 1)} \times 16^{k tw(G)} \times k^{4 tw(G)})$ where $n$ is the number of phases; $L_1^{k} \times L_2^{k}$ is the number of of boxes, $L_1^{k} \times L_2^{k} \times k \times tw(G) \times log(n)$ is the the processing time to compute if the states are the same; and $n^{2(k tw(G) + 1)} \times 16^{k tw(G)} \times k^{4 tw(G)}$ is the number of distincts combinatorial frontiers.



Now that we have provided a bi-FPTAS for graphs of bounded tree-width, it would be interesting to look at graphs bounded by more generic graph parameters like the clique-width and local tree-width. The latter is all the more interesting as we know that planar graphs have locally bounded tree-width and can be used to model numerical simulations on HPC architectures.
\bibliographystyle{splncs03}

%
%
%
%
\end{document}